\documentclass[11pt,letter]{article}
\usepackage[hidelinks]{hyperref} 
\usepackage{fullpage}
\usepackage[utf8x]{inputenc}
\usepackage{amsthm}
\usepackage{todonotes,wrapfig,float,graphicx,amssymb,textcomp,array,amsmath}
\usepackage{enumerate,enumitem}
\usepackage{multirow}
\usepackage{tabularx}
\usepackage{color,xcolor}

\definecolor{mycolor}{rgb}{0, 0, 0}
\usepackage{todonotes}
\usepackage[titletoc,title]{appendix}

\usepackage{lineno}

\newcommand{\revm}[1]{{\color{black} #1}}
\newcommand{\rev}[1]{{\color{black} #1}}
\newcommand{\rv}[1]{{\color{black} #1}}
\newcommand{\rvm}[1]{{\color{black} #1}}
\newcommand{\old}[1]{{}}

\title{Noncrossing Longest Paths and Cycles\thanks{This work has been presented at the \emph{32nd International Symposium on Graph Drawing and Network Visualization (Vienna, 2024)}, GD~2024. The main results and ideas have also been reported at \url{https://11011110.github.io/blog/2024/09/25/long-non-crossing.html}.}}

\author{
Greg Aloupis\thanks{Khoury College of Computer Sciences, Northeastern University, Boston, MA, USA, \texttt{g.aloupis@northeastern.edu}} 
\and Ahmad Biniaz\thanks{Corresponding author. School of Computer Science, University of Windsor, Windsor, ON, Canada, \texttt{abiniaz@uwindsor.ca}. Research supported by NSERC.}
\and Prosenjit Bose\thanks{School of Computer Science, Carleton University, Ottawa, ON, Canada, \texttt{jit@scs.carleton.ca, anil@scs.carleton.ca, michiel@scs.carleton.ca}. Research supported by NSERC.} 
\and Jean-Lou De Carufel\thanks{School of Electrical Engineering and Computer Science, University of Ottawa, Ottawa, ON, Canada, \texttt{jdecaruf@uottawa.ca, saeedodak@gmail.com}. Research supported by NSERC.} 
\and David Eppstein\thanks{Computer Science Department, University of California, Irvine, CA, USA, \texttt{eppstein@uci.edu}. Research supported by NSF grant CCF-2212129.} 
\and Anil Maheshwari\footnotemark[4]
\and Saeed Odak\footnotemark[5]
\and Michiel Smid\footnotemark[4]
\and Csaba D. T\'{o}th\thanks{Department of Mathematics, California State University Northridge, Los Angeles, CA; and Department of Computer Science, Tufts University, Medford, MA, USA, \texttt{csaba.toth@csun.edu}. Supported by NSF grant DMS-2154347.} 
\and Pavel Valtr\thanks{Department of Applied Mathematics, Charles University, Prague, Czech Republic, \texttt{valtr@kam.mff.cuni.cz}. {Research supported by Czech Science Foundation grant GA\v CR
23-04949X.}} }

\date{}
\newtheorem{lemma}{Lemma}
\newtheorem*{lemma*}{Lemma}
\newtheorem{corollary}{Corollary}
\newtheorem{proposition}{Proposition}
\newtheorem{conjecture}{Conjecture}
\newtheorem{theorem}{Theorem}
\newtheorem{observation}{Observation}
\newtheorem*{problem*}{Problem}
\newtheorem*{claim*}{Claim}
\newtheorem*{invariant*}{Invariant}
\newtheorem{question}{Question}
\newtheorem*{remark}{Remark}

\begin{document}
\maketitle
\begin{abstract}
Edge crossings in geometric graphs are sometimes undesirable as they could lead to unwanted situations such as collisions in motion planning and  inconsistency in VLSI layout. 
Short geometric structures such as shortest perfect matchings, shortest spanning trees, shortest spanning paths, and shortest spanning cycles on a given point set are inherently noncrossing.
However, the longest such structures need not be noncrossing. 
In fact, it is intuitive to expect many edge crossings in various geometric graphs that are longest.

Recently, \'{A}lvarez-Rebollar, Cravioto-Lagos, Mar\'{\i}n, Sol\'{e}-Pi, and Urrutia (Graphs and Combinatorics, 2024) constructed a set of points for which the longest perfect matching is noncrossing.
They raised several challenging questions in this direction. In particular, they asked whether the longest spanning path, on \rv{every} 
finite  
set of points in the plane, must have a pair of crossing edges.  They also conjectured that the longest spanning cycle must have a pair of crossing edges.  

In this paper, we give a negative answer to the question and also refute the conjecture. We present a framework for constructing arbitrarily large point sets for which the longest perfect~matchings, the longest spanning paths, and the longest spanning cycles are noncrossing.
\end{abstract}

\section{Introduction}

Traversing points in the plane by a polygonal path or cycle possessing a desired property has a rich background. For instance, the celebrated travelling salesperson problem asks for a polygonal path or cycle with minimum total edge length \cite{Arora1998,Mitchell99,Papadimitriou1977}. In recent years, there has been  increased interest in paths and cycles with properties such as being noncrossing \cite{Aichholzer2010,Cerny2007}, minimizing the longest edge length~\cite{AnKS21,Biniaz2022,KaoS09}, maximizing the shortest edge length \cite{Arkin1999}, minimizing the total or largest turning angle \cite{Aggarwal1999,Biniaz24,Dumitrescu12,Fekete1997}, and minimizing the number of turns
\cite{Biniaz24b,Dumitrescu2014,Stein2001} to name a few. \rv{Finding a} longest cycle---the MaxTSP \rv{problem}---is NP-hard in Euclidean spaces of dimension 
\rv{three or higher,} 
but the complexity of the \revm{MaxTSP problem in the plane} is unknown~\cite{Barvinok_2007,Barvinok2003,DBLP:conf/soda/Fekete99}.\footnote{It is interesting that the MaxTSP under the $L_1$-norm for points in the plane can be solved in linear time \cite{Barvinok2003,DBLP:conf/soda/Fekete99}, in contrast to the fact that MinTSP in this case is NP-hard~\cite{Garey1976}.} Paths and cycles that have combinations of these properties have also attracted attention. For example, 
\rv{maximizing the total edge length subject to the constraint that the edges are noncrossing}~\cite{Alon1995,Dumitrescu2010} is difficult: 
to achieve a larger length we typically introduce more crossings.

Edge crossings in geometric graphs are usually undesirable as they have the potential of creating unwanted situations such as collisions in motion planning and inconsistency in VLSI layout.
 They are also undesirable in the context of graph drawing and network visualization as they make drawings more difficult to read and use.
Short\rv{est} geometric structures such as shortest perfect matchings, shortest spanning trees, shortest spanning paths, and shortest spanning cycles are inherently noncrossing. This property, however, does not necessarily hold if the structure is not shortest. For long structures such as longest perfect matchings, longest spanning trees, longest spanning paths, and longest spanning cycles---the other end of the spectrum---it seems natural to expect many crossings. Counting crossings in geometric graphs and finding geometric structures with a minimum or maximum number of crossings are active research areas in discrete geometry. \rv{This line of research is also related \revm{to} geometric Ramsey-type problems~\cite{Aronov1994,Karolyi2013}. For example, it is known that every} set of $n$ points in the plane in general position admits a {\em crossing family} \rv{(a matching with pairwise crossing edges) of size $n^{1-o(1)}$~\cite{Pach2021}, but it remains open whether this bound can be improved to $\Omega(n)$~\cite{BMP05}}. 

The noncrossing property of shortest structures is mainly ensured by the triangle inequality. The triangle inequality also implies that the longest structures often have crossings because a structure usually gets longer by creating more crossings~\rv{\cite{Alon1995}}. \rv{There are 
approximation algorithms for the problem of finding longest noncrossing structures (such as matchings, paths, cycles, and trees) \cite{Alon1995,BiniazBCCEMS19,Cabello0KMT25,Dumitrescu2010}}. Along this direction, one might wonder whether a longest structure (defined on an arbitrarily large point set) is necessarily crossing. This was explicitly asked \rv{by Alvarez-Rebollar, Cravioto-Lagos, Mar\'{\i}n, Sol\'e-Pi, and Urrutia~\cite{Rebollar2024}.} Among other interesting results, they 
presented
arbitrarily large planar point sets for which the longest perfect matching is noncrossing. They asked the following question and proposed the following conjecture:
\begin{question}[\rv{\cite{Rebollar2024}}]
\label{path-question}
 For every sufficiently large planar point set, must the longest spanning path have two edges that cross each other?
\end{question} 
\begin{conjecture}[\rv{\cite{Rebollar2024}}]
\label{cycle-conjecture}
    The longest spanning cycle on every 
    sufficiently large 
    set of points in the plane has a pair of crossing edges.
\end{conjecture}

The ``sufficiently large'' condition in the question and conjecture \rv{cannot be dropped,}
as otherwise one can take any 3 points in general position, or any 4 points that are not in a convex position---for such point sets, all spanning paths and cycles are noncrossing.

In the other direction, one might wonder about maximizing the number of crossings in cycles. 
Let $C(n)$ be the largest number such that \rv{every} set of $n$ points in the plane admits a spanning cycle with at least $C(n)$ pairs of crossing edges. 
The following lower and upper bounds are known~\rv{\cite{Rebollar2024,Rebollar2015}}: $\frac{1}{12}\, n^2-O(n)<C(n)<\frac{5}{18}\, n^2-O(n)$. In other words, \rv{every} set of $n$ points in the plane admits a spanning cycle with at least $\frac{1}{12}\, n^2-O(n)$ crossings, and \rv{there are arbitrarily large point sets that do} not admit any cycle with more than $\frac{5}{18}\, n^2-O(n)$ crossings. 

\subsection{Our Contributions}
In this paper, we provide a negative answer to Question~\ref{path-question} and \rv{refute} Conjecture~\ref{cycle-conjecture}.  For \rv{every} integer $n\ge 1$, we present a set of $n$ points in the plane for which the longest spanning path is unique and noncrossing. Similarly, for \rv{every integer $n\ge 3$}, we present a set of $n$ points in the plane for which the longest spanning cycle is unique and noncrossing. To build such point sets, we use the following framework: First, we choose a set $P$ of points on the $x$-axis for which the longest structure may not be unique. Then, we assign new $y$-coordinates to points in $P$ to obtain a new point set $P'$ for which the longest structure corresponds to one in $P$ and is \rv{both} unique and noncrossing.

\rv{The paper is organized as follows. In Section~\ref{sec-DimensionOne} we prove structural properties of longest paths and cycles on the real line. We employ these properties and construct noncrossing longest paths in Section~\ref{path-section}, noncrossing longest cycles in Section~\ref{cycle-section}, and noncrossing longest matchings in Section~\ref{matching-section}. In Section~\ref{properties-sec}, we present some structural properties of longest paths and cycles in the plane, which may be of independent interest.}

\subsection{Preliminaries}
 {\color{mycolor}
All point sets considered in this paper are in the Euclidean plane.
A {\em geometric graph} is a graph with vertices represented by points and edges represented by line segments between the points.} Let $P$ be a finite point set. 
A {\em spanning path} for $P$ is a path drawn  with straight-line edges such that every point in $P$ lies at a vertex of the path and every vertex of the path lies at a point in $P$. A {\em spanning cycle} is defined analogously. In other words, a spanning path is a Hamiltonian path in the complete geometric graph on $P$, and a spanning cycle is a Hamiltonian cycle in this graph.

Consider two line segments, each connecting a pair of points in $P$. If the \rv{relative} interiors of the segments intersect, then we say that they {\em cross};  this configuration is called a {\em crossing}. A path or a cycle is called {\em noncrossing} if \rv{no two edges cross}. 
We say that an edge $e$ {\em intersects the $y$-axis} if the intersection of $e$ and the $y$-axis is not empty (the intersection could be an endpoint of $e$). 
{\color{mycolor}We denote the undirected edge between two points $p$ and $q$ by $pq$, the directed edge from $p$ towards $q$ by $(p,q)$, and the Euclidean distance between $p$ and $q$ by $|pq|$.}  \rv{The {\em length}  of an edge with endpoints $p$ and $q$ is the Euclidean distance between $p$ and $q$.} The {\em length} of a geometric graph $G$ is the sum of the lengths of its edges, and we denote it by $|G|$. \rv{In general, if $w:E\to \mathbb{R}$ is a weight function on the edges of a graph $G=(V,E)$, then the weight of a subgraph $H$ of $G$ is the sum of the weights of its edges: $w(H)=\sum_{e\in E(H)} w(e)$.}

\rv{
The following lemma and its corollary are simple but play an important role in our constructions for noncrossing longest paths and cycles.
\begin{lemma}\label{epsilon-lemma-revised}
    Let $G=(V,E)$ be a graph with two edge weight functions $w_1:E\to \mathbb{R}$ and $w_2:E\to \mathbb{R}$ such that for every $e\in E$, we have $|w_1(e)-w_2(e)|\leq 2\, \alpha$. 
    If $\mathcal{A}$ and $\mathcal{B}$ are two families of subgraphs of $G$ with $m$ edges such that $\beta:= \min\{ w_1(A)-w_1(B) : A\in \mathcal{A}, B\in \mathcal{B}\}>0$, then $\min\{w_2(A)-w_2(B) : A\in \mathcal{A}, B\in \mathcal{B}\}\geq \beta-4m\alpha$.
\end{lemma}
\begin{proof}	
    For every $A\in \mathcal{A}$ and $B\in \mathcal{B}$, 
    we have 
\begin{align*}
    w_2(A)-w_2(B) 
    &= \sum_{e\in E(A)} w_2(e) - \sum_{f\in E(B)} w_2(f)\\
    &\geq  \sum_{e\in E(A)} \big(w_1(e) - |w_2(e)-w_1(e)| \big)
        - \sum_{f\in E(B)} \big(w_1(f) + |w_2(f)-w_1(f)| \big)\\
    &\geq \sum_{e\in E(A)} w_1(e) - \sum_{f\in E(B)}w_1(f)  -4m\alpha\\
    &=w_1(A)-w_1(B)-4m\alpha.
\end{align*}
Consequently, we have $\min\{w_2(A)-w_2(B) : A\in \mathcal{A}, B\in \mathcal{B}\}\geq \beta-4m\alpha$.
\end{proof}
\begin{corollary}\label{epsilon-cor-revised}
    Let $P$ be a set of $n$ points in the plane and $\mathcal{S}$ a family of geometric graphs on $P$ with at most $m$ edges (for example, spanning cycles with $m=n$ or spanning paths with $m=n-1$). Let $\mathcal{S}=\mathcal{A}\cup \mathcal{B}$ be a partition such that $\mathcal{A}\neq \emptyset$ and $\beta:= \min\{|A|-|B| : A\in \mathcal{A}, B\in \mathcal{B}\}>0$. Assume that every point in $P$ is perturbed by a distance of at most $\alpha$, $0\le \alpha < \frac{\beta}{4m}$, and let $\mathcal{S}'$ be the set of perturbed graphs corresponding to $\mathcal{S}$. Then every longest graph in $\mathcal{S}'$ corresponds to a graph in $\mathcal{A}$. 
\end{corollary}
\begin{proof}	
We use Lemma~\ref{epsilon-lemma-revised} with the following setup: Let $G=(P,E)$ be the complete graph on $P$. For every edge $pq\in E$, let $w_1(pq)=|pq|$ and $w_2(pq)=|p'q'|$, where $p$ is perturbed to $p'$ and $q$ to $q'$. If every point is perturbed by a distance of at most $\alpha$, then for every edge $pq \in E$ the triangle inequality yields $w_2(pq)=|p'q'| \leq |pp'|+|pq|+|qq'|\leq w_1(pq)+2\alpha$ and $w_2(pq)=|p'q'| \geq |pq|-|pp'|-|qq'|\geq w_1(pq)-2\alpha$, hence $|w_2(pq)-w_1(pq)|\leq 2\alpha$.  

Now consider a longest graph $G'\in \mathcal{S}'$, corresponding to a graph $G\in \mathcal{S}$. We have $|G'|=w_2(G)$ and $|G|=w_1(G)$. If $G\notin \mathcal{A}$, then for every graph $A\in \mathcal{A}$, Lemma~\ref{epsilon-lemma-revised} gives $w_2(A)-w_2(G)\geq \beta-4\alpha m>0$,
which means that $G'$ is not a longest graph in $\mathcal{S}'$: a contradiction. We conclude that $G\in \mathcal{A}$, as required. 
\end{proof}
}

\section{Longest Paths and Cycles on the Real Line}
\label{sec-DimensionOne}
In this section we characterize the longest paths and cycles in dimension one. These observations play a pivotal role in our constructions in the plane (Sections~\ref{path-section} and~\ref{cycle-section}). We say that an edge $e$ {\em intersects a point $p$} if the intersection of $e$ and $p$ is not empty (the intersection could be an endpoint of $e$). For a sorted set of $2k{+}1$ numbers, the median is the number \rv{at position} $k{+}1$, and for a sorted set of $2k$ numbers, the median is the mean of the two numbers \rv{at positions} $k$ and $k{+}1$.

\begin{lemma}\label{endpoint-lemma} 
Let $P$ be a set of \rv{$n\geq 2$} points in $\mathbb{R}$, \rv{and let $H$ be a longest spanning path on $P$. Then.  
	\begin{itemize}\setlength{\itemindent}{2em}
		\item [$(i)$] the two endpoints cannot lie on strictly the same side of the median of $P$, and
		\item [$(ii)$] every edge of $H$ intersects the median of $P$.
	\end{itemize} 
}
\end{lemma}
\begin{proof}
Let $P=\{p_1,\ldots , p_n\}$ and assume \rv{without loss of generality} that 0 is the median of $P$\rv{. Note that $0\in P$ if $n$ is odd and $0\notin P$ if $n$ is even}. We prove both statements (i) and (ii) by contradiction. 

Suppose that (i) does not hold. Orient the edges of $H$ to make it a directed path. 
Let $p_s$ and $p_e$ be the starting and ending points of $H$, respectively. For the sake of contradiction, assume that $p_s$ and $p_e$ have the same sign, which we may assume, due to symmetry, to be positive, \rv{that is,} 
$p_s,p_e >0$. Then, the sum of degrees of vertices in $H$ to the left of the origin is 2 more than the sum of degrees of vertices to the right. Therefore, $H$ must have a directed edge $(p_{a}, p_{b})$ where \rv{$p_a,p_b \leq 0$}. 
If $p_b<p_a$, then by replacing 
\rv{the undirected edge $p_ap_b$ with the edge $p_bp_s$} 
we obtain a longer undirected path; and if $p_b>p_a$ by replacing 
$p_ap_b$ with $p_ep_a$ we obtain a longer undirected path. Both cases lead to a contradiction.

Suppose that (ii) does not hold.  Let $(p_a,p_b)$ be an edge of $H$ that does not intersect the median. Due to symmetry, \rv{we may} assume that $p_a,p_b<0$. Part~(i) implies that the sum of degrees on negative and positive vertices differ by at most one (and must be equal when $n$ is even). Thus $H$ must have an edge $(p_c,p_d)$ such that \rv{$p_c,p_d \geq 0$}. By replacing these edges with $p_ap_c$ and $p_bp_d$ we obtain an (undirected) spanning path that is longer than $H$ because $|p_a-p_c|+|p_b-p_d|>|p_a-p_b|+|p_c-p_d|$. This contradicts $H$ being a longest path. 
\end{proof}

\rv{We can characterize longest spanning paths on an even number of points.}
\begin{lemma}\label{path-lemma} 
Let $P$ be a set with an even number of points in $\mathbb{R}$. Let $H$ be a spanning path on $P$. Then $H$ is a longest spanning path if and only if
	\begin{itemize}\setlength{\itemindent}{2em}
		\item [$(i)$] every edge of $H$ intersects the median of $P$, and
		\item [$(ii)$] the two endpoints of $H$ are the two points closest to the median of $P$.
	\end{itemize} 
\end{lemma}
\begin{proof}
Let $P=\{p_1,\ldots , p_n\}$ so that $p_i<p_j$ for all $i,j\in\{1,\ldots , n\}$, $i<j$, and assume \rv{without loss of generality} that 0 is the median of $P$. Note that $0\notin P$ since $n$ is even. First, we prove that if $H$ is a longest spanning path on $P$, then (i) and (ii) hold. Property (i) follows from Lemma~\ref{endpoint-lemma}(ii). We prove (ii) by contradiction. 

Suppose that (ii) does not hold: without loss of generality $p_{n/2}$ is not an endpoint of $H$. (The case for $p_{n/2+1}$ can be handled symmetrically).  Then $H$ has an endpoint $p_a$ with $a < n/2$. Orient the edges of $H$ so that the path is directed from $p_a$ towards the other endpoint. Let $(p_{n/2}, p_b)$ be the outgoing edge from $p_{n/2}$. By property~(i), we have $p_b>0$. By removing the undirected edge $p_{n/2} p_b$ we obtain two undirected paths, and $p_b$ is an endpoint on one of those paths. Next, join the paths with a new edge $p_ap_b$. Thus we obtain an spanning path that is longer than $H$ because $|p_a-p_b|>|p_{n/2}-p_b|$. This contradicts $H$ being longest.

Finally, we prove that every spanning path $H$ that satisfies (i) and (ii) is longest, using a direct proof. Consider a longest spanning path $L$ on $P$. By the necessity proof, (i) and (ii) hold for~$L$. This implies that the 
interval $[p_{n/2},p_{n/2+1}]$ is contained in each of the $n{-}1$ edges, hence it contributes to the length of~$L$ with multiplicity $n-1$. Similarly, for \rv{every} $i\in \{1,2,\ldots , n/2-1\}$ the 
intervals $[p_{n/2-i},p_{n/2-i+1}]$ and $[p_{n/2+i},p_{n/2+i+1}]$ each contribute to the length of $L$ with multiplicity $n+2-2i$. 
%
%
 See Figure~\ref{one-dimension-fig}. 
On the other hand, \rv{every} spanning path (including $H$) that satisfies (i) and (ii) receives the exact same multiplicities from the corresponding intervals. Therefore $H$ and $L$ have the same length, and hence $H$ is also a longest path.	
\end{proof}

\begin{figure}[!ht]
	\centering
	\setlength{\tabcolsep}{0in}
	\includegraphics[width=.65\columnwidth]{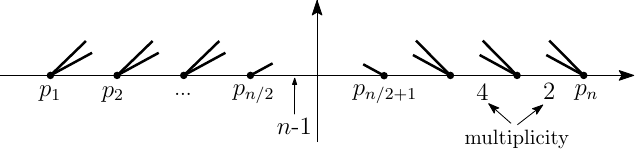}
	\caption{Illustration of a longest path for a point set on a line, for the case where the number of points, $n$, is even. Numbers below intervals $[p_{n/2+i},p_{n/2+i+1}]$ represent the multiplicity of the contribution of the corresponding intervals to the length of the longest path.}
		\label{one-dimension-fig}
\end{figure}

For an odd number of points, a similar argument yields essentially the same characterization as Lemma~\ref{path-lemma}. 

\rv{
\begin{lemma}
    \label{odd-path-cor}
    Let $P$ be a set with an odd number of points in $\mathbb{R}$. Let $H$ be a spanning path on $P$. Then $H$ is a longest spanning path if and only if
\begin{itemize}\setlength{\itemindent}{2em}
		\item [$(i)$] every edge of $H$ intersects the median of $P$, and
		\item [$(ii)$] one endpoint of $H$ is the median and the other endpoint is a closest point to the median.
	\end{itemize} 
\end{lemma}
\begin{proof}
Let $n=2k+1$ and $P=\{p_1,\ldots , p_n\}$ so that $p_i<p_j$ for all $i,j\in\{1,\ldots , n\}$, $i<j$. Note that the median of $P$ is $p_{k+1}$, and assume \rv{without loss of generality} that $p_{k+1}=0$. We prove that if $H$ is a longest spanning path on $P$, then (i) and (ii) hold. Property (i) follows from Lemma~\ref{endpoint-lemma}(ii). We prove (ii) by contradiction. 
	
First, suppose that the median $p_{k+1}=0$ is not an endpoint of $H$. By Lemma~\ref{endpoint-lemma}(i), the sum of degrees on (strictly) positive and negative vertices are the same. Consequently, the median is incident to a vertex on each side. Without loss of generality, $H$ contain the edges $(p_a, 0)$ and $(0,p_b)$ where $p_a<0<p_b$. Replace the undirected edges $p_a0$ and $0p_b$ with $p_a p_b$ and $p_e 0$, where $p_e$ is the endpoint of $H$. Since $|p_a|+|p_b|=|p_b-p_a|<|p_b-p_a|+|p_e|$, we obtain a spanning path longer than $H$, contradicting the maximality of $H$.

Second, suppose that the endpoints of $H$ are $p_a$ and $p_{k+1}=0$, but there is another point in $P$ between $p_a$ and $0$. Assume without loss of generality that $p_a<p_b<0$. Orient the edges of $H$ so that the path is directed from $p_a$ towards the other endpoint $0$. Let $(p_b, p_c)$ be the outgoing edge from $p_b$. By property~(i), we have $p_c\geq 0$. Replace the undirected edge $p_b p_c$ with $p_a \revm{p_c}$, to obtain a new spanning path on $P$, which is longer than $H$ since $|p_b-p_c|<|p_a-p_c|$, contradicting the maximality of $H$.

Finally, suppose that the endpoints of $H$ are $p_a$ and $p_{k+1}=0$, there is no point in $P$ between $p_a$ and $0$, but the closest point to $0$ is $p_b$ with $p_b\neq p_a$. Without loss of generality, assume $p_a<0<p_b$. Orient the edges of $H$ so that the path is directed from $p_a$ towards the other endpoint $0$. Let $(p_b, p_c)$ be the outgoing edge from $p_b$, and $(p_d,0)$ the last edge of $H$. By property~(i), we have $p_c\leq 0$ and $p_d>0$. 
If $p_b p_c = p_d 0$, then replace this edge $p_b p_c = p_d 0 = p_b 0$ with $p_a 0$ to obtain a new spanning path on $P$, which is longer since $|p_b|<|p_a|$.
Otherwise $p_b p_c\neq p_d 0$, which implies under our assumptions that $p_c<p_a<0<p_b<p_d$. In this case, replace the undirected edges $p_b p_c$ and $p_d 0$ with $p_a p_d$ and $p_c 0$, to obtain a new spanning path on $P$, which is longer than $H$ because $|p_b-p_c| + |p_d| = |p_a-p_c| + |p_a| + 2\, |p_b| + |p_d-p_b| < 
|p_a-p_c| + \revm{2\, |p_a|} + |p_b| + |p_d-p_b| = |p_a-p_d| + |p_c|$, contradicting the maximality of $H$.

The \revm{sufficiency} of conditions (i) and (ii) can be proved similarly to the proof of Lemma~\ref{path-lemma}: comparing the contributions of the intervals $[p_i,p_{i+1}]$, $i=1,\ldots ,n$, \revm{in} a longest spanning path $L$ on $P$, and a spanning path satisfying (i) and (ii). If $p_k$ and $p_{k+1}$ are the endpoints of $H$, then the intervals $[p_k, p_{k+1}]$ and $[p_{k+1},p_{k+2}]$ contribute to the length of $H$ with multiplicity \revm{$n-2$ and $n-1$}, respectively. In the special case that the median has two closest points in $P$ (that is, $|p_k|=|p_{k+2}|$), then
$p_{k+2}$ may be an endpoint of $H$, and then these intervals contribute \revm{$n-1$ and $n-2$}, respectively. For \rv{every} $i\in \{1,\ldots , k-1\}$, however, the intervals $[p_{k-i},p_{k-i+1}]$ and $[p_{k+i+1},p_{k+i+2}]$ each contribute with multiplicity $2(k-i)$. 
\end{proof}	
}

\rv{We can also characterize longest spanning cycles on $n$ points on a line.}

\begin{lemma}\label{cycle-lemma} 
Let $P$ be a finite set in $\mathbb{R}$. 
\begin{itemize}\setlength{\itemindent}{2em}
		\item [$(i)$] A spanning cycle on $P$ is longest \rv{if and only if}
  each of its edges intersects the median of~$P$. 
		\item [$(ii)$] If $P$ contains an odd number of points, then for \rv{every} longest spanning cycle the two edges incident to the median lie on opposite sides of it. 
   \item [$(iii)$] Assume that $P$ contains $n=2k{+}1$ points and there is an interval $I$ of length $h>0$ between the leftmost $k{+}1$ and the rightmost $k$ points. Then in \rv{every} longest spanning cycle, $n{-}1=2k$ edges contain the interval $I$; and if a spanning cycle has fewer than $2k$ edges that contain $I$, then it is at least $2h$ shorter than a longest cycle.
\end{itemize}
\end{lemma}
\begin{proof}
Let $P=\{p_1,\ldots , p_n\}$ so that $p_i<p_j$ for all $i<j\in\{1,\ldots , n\}$, and assume \rv{without loss of generality} that 0 is the median of $P$.
Note that $0\notin P$ if $n$ is even, and $p_{\lceil n/2\rceil}=0$ if $n$ is odd.

First we prove the \revm{necessity} of (i) by contradiction. Let $C$ be a longest cycle on $P$, and orient its edges to obtain a directed cycle. Suppose, for the sake of contradiction, that the edge $(p_a,p_b)$ of $C$ does not intersect the median. We may assume \rv{without loss of generality} that $p_a,p_b<0$. The sum of vertex degrees  strictly on the left and right side of the median are the same, and the edges that contain $0$ in their interior contribute 1 to both sums. Consequently, either $C$ contains an edge $(p_c,p_d)$ with $p_c,p_d>0$; or (when $n$ is odd) the two edges incident to the median, say $(p_c,0)$ and $(0,p_d)$ satisfy $p_c,p_d>0$. 
In the first case, we can replace the undirected edges $p_a p_b$ and $p_c p_d$
with $p_a p_c$ and $p_b p_d$. 
In the second case, replace $p_a p_b$ and $p_c 0$ with $p_a p_c$ and $p_b 0$. In both cases, we obtain a longer (undirected) spanning cycle, contradicting the maximality of $C$. 

The \revm{sufficiency} of (i) can be proved by a counting argument similar to that of Lemma~\ref{path-lemma}.

Next, we prove (ii) by contradiction. 
Suppose that the median, $0\in P$, is incident to the edges $(p_c,0)$ and $(0,p_d)$ with $p_c,p_d>0$. \rv{As the sum of vertex degrees  strictly on the left and right side of the median are the same,} there is \rv{an edge of $C$ strictly to the left of $0$; this contradicts part (i).} 


To prove the first claim in (iii), note that if $n=2k{+}1$, then the median is the $(k{+}1)$-st point in $P$, that we denote by $p_0$. Let $C$ be a longest cycle on $P$. 
Now (i) and (ii) imply that exactly one edge of $C$ (which is incident to $p_0$) does not contain $I$. The remaining $n{-}1=2k$ edges contain $I$.

For the second claim in (iii), let $C$ be a spanning cycle on $P$ in which fewer than $2k$ edges contain $I$. Orient the edges of $C$ to obtain a directed cycle. The sum of degrees of the leftmost $k{+}1$ (\rv{respectively}, rightmost $k$) vertices is $2k{+}2$ (\rv{respectively}, $2k$), and the edges containing $I$ have fewer than $2k$ left (\rv{respectively}, right) endpoints. Consequently, the leftmost $k{+}1$ (\rv{respectively}, rightmost $k$) points in $P$ induce at least two edges (\rv{respectively}, one edge) of $C$. Therefore, $C$ contains two edges, $(p_a,p_b)$ and $(p_c,p_d)$, such that $p_a,p_b$ are to the left of $I$ and \revm{$p_c,p_d$} are to the right of $I$. We can replace these two edges
with $(p_a,p_c)$ and $(p_b,p_d)$, to obtain a spanning cycle $C'$ that traverses $I$ two more times than $C$. In particular, we have $|C'|\geq |C|+2\, |I|=|C|+2h$,
hence $|C|\leq |C'|-2h\leq |C_{\max}|-2h$, where $C_{\max}$ is a longest cycle on $P$.
\end{proof}

\section{Noncrossing Longest Paths}
\label{path-section}

Let $n\ge 1$ be an integer. In this section, we construct 
$n$ points for which the longest spanning path is unique and noncrossing. This can be easily observed for $n<5$: For example, for $n=4$, \rv{every} spanning path of the vertices of a triangle and a point in the interior is noncrossing.  Thus, we will now assume that $n\ge 5$. 
{\color{mycolor} 
In Section~\ref{sec-DimensionOne}, we uncovered 
some structural properties of longest paths for $n$ points 
on a line. Here we show how to construct a 2-dimensional point set starting with $n$ points on the $x$-axis and then  assigning $y$-coordinates to the points.
We show that} the longest path is unique and noncrossing. We describe our construction for the case where $n$ is even; the construction for the case where $n$ is odd follows with some minor changes. The following theorem summarizes the main result in this section. 

\begin{theorem}\label{thm:path}
    For every integer $n\ge 1$ there exists a set of $n$ points in the plane for which the longest spanning path is unique and noncrossing. 
\end{theorem}

\old{
{\color{mycolor}In Section~\ref{even-path-overview} we give an overview of our construction for an even number of points. The details and proofs are given in Section~\ref{even-path-details}. The case of odd paths is considered in Section~\ref{odd-path}.}
}

\rvm{
In Section~\ref{even-path-details} we prove a generalization of Theorem~\ref{thm:path}. The case of odd paths is handled in Section~\ref{odd-path}.
}

\old{
\subsection{A Path with an Even Number of Points: An Overview}
\label{even-path-overview}

For $k\geq 3$, consider a set $P$ of $n=2k$ points $p_i$ on the $x$-axis such that $p_1=(0,0)$ and $p_i=(i,0)$ for $ i=-1,\pm2,\dots,\pm k$, as illustrated in Figure~\ref{path-even-fig}(a). 
The longest spanning path for this point set is not unique. In fact, {\color{mycolor}Lemma~\ref{path-lemma} implies that} \rv{every} spanning path with endpoints $p_1$ and $p_{-1}$ and with all edges crossing the $y$-axis is a longest path. {\color{mycolor} Conversely, \rv{every} longest path must have endpoints $p_1$ and $p_{-1}$, and its edges must cross the $y$-axis.} 
Let $\cal H$ be the set of these paths. 
Let $P'$ be the point set obtained by assigning to each point $p_i$ a $y$-coordinate $y_i$ such that, as illustrated in Figure~\ref{path-even-fig}(b), the following holds:
\[\frac{1}{8k}=y_1\gg y_{-2}\gg y_2\gg y_{-3}\gg y_3\gg \cdots\gg  y_{-k}\gg y_k\gg y_{-1}=0.\]

The value $y_1$ is much larger than $y_{-2}$, which is in turn much larger than $y_2$ and so on. Notice that the largest $y$-coordinate $y_1$ is \rv{$\frac{1}{8k}$} which is much smaller than $1$. Due to the small $y$-coordinates, a longest path $H'$ on $P'$ corresponds to a path $H\in \cal H$. The length of $H'$ is roughly the length of $H$ plus a very small value \rv{$Y(H')$}, which depends on the new $y$-coordinates. Let $e_1$ be the only edge of $H'$ incident to $p_1$. Since $p_1$ has a very large $y$-coordinate compared to other points, the contribution of $e_1$ to \rv{$Y(H')$} is larger than the contribution of other edges. The contribution of $e_1$ is maximized if it connects $p_1$ to the nearest plausible neighbor, which is $p_{-2}$; this can be observed from Figure~\ref{path-even-fig}(b).  Therefore $e_1=p_1p_{-2}$. By a similar argument, $p_{-2}$ gets connected to $p_2$, and so on. It follows that the path $H'$ is unique and it is $p_1,p_{-2},p_2,p_{-3},p_3,\dots, p_{-k},p_k,p_{-1}.$
This path is $y$-monotone, and hence noncrossing; see Figure~\ref{path-even-fig}(b).
}

\subsection{A Path with an Even Number of Points}
\label{even-path-details}

\rv{
In this section, we prove a stronger statement than Theorem~\ref{thm:path} for even $n$, and show that \revm{a} set $P$ of an even number of points on the $x$-axis can be perturbed, by slightly increasing their $y$-coordinates, to a point set $P'$ for which the longest spanning path is unique and $y$-monotone (hence noncrossing).

\begin{lemma}\label{evenpath-lem+}
 For every even integer $n\ge 2$ \revm{and} every set $P$ of $n$ real numbers, 
 there exists a set $P'$ of $n$ points in the plane with the following properties:
\begin{enumerate}
    \item the $x$-projection of $P'$ is $P$;
    \item the $x$-projection of \rv{every} longest path on $P'$ is a longest path on $P$;
    \item the longest spanning path on $P'$ is unique and noncrossing.
\end{enumerate}
\end{lemma}

\revm{The rest of this subsection is dedicated to the proof of Lemma~\ref{evenpath-lem+}}. The statement is trivial for $n=2$, so we may assume $n\geq 4$.
Let $P=\{x_{-1},x_1,x_{-2},x_2,\ldots , x_{-k},x_k\}$ for some integer $k=n/2$ such that $x_i<x_j$ if and only if $i<j$. Furthermore, we may assume without loss of generality that the median of $P$ is 0. Since $n$ is even, then $0\notin P$ and we have $x_i<0$ if and only of $i<0$. 

We construct the point set $P'=\{p_{-1},p_1,p_{-2},p_2,\ldots , p_{-k},p_k\}$ where $p_i=(x_i,y_i)$ for some $y_i\geq 0$ to be specified later. In particular, the $y$-coordinates will satisfy  
\[
    y_1> y_{-2}> y_2> y_{-3}> y_3> \cdots >y_{-k}> y_k> y_{-1} = 0.
\]
\rev{First we choose $y_1$.} For every $Q\subseteq P$ \rev{with $|Q|\ge 3$}, let $\beta(Q)$ be the difference between the length of the longest and the second longest spanning paths on $Q$; and let $\beta=\rev{\min}_{Q\subseteq P} \beta(Q)$. We set $y_1=\frac{\beta}{8k}$. 
Since $P$ is a set of integers, $\beta\geq 1$. 
Then Corollary~\ref{epsilon-cor-revised} implies that the \rev{$x$-projection of every longest path on $P'$ is a longest path on $P$}.
We choose the remaining $y$-coordinates in $k-1$ iterations and maintain the following invariant for $i\in \{1,\ldots , k\}$. 

\begin{invariant*}[$M_i$]
There exist real numbers $\frac{\beta}{8k}=y_1> y_{-2}> y_2>  \cdots >y_{-i}> y_i> y_{-1} = 0$ \rev{such that} for each $i\in \{1, \ldots, k\}$, every longest spanning path on $P'$ \rev{starts with} 
    $p_1 ,p_{-2} ,p_2 ,\dots ,p_{-i} ,p_i$ as a subpath.
\end{invariant*}

The Invariant~$(M_1)$ vacuously holds as the 1-vertex path $p_1$ is a subpath of every spanning path on $P'$.}

\begin{figure}[!ht]
	\centering
	\setlength{\tabcolsep}{0in}
	\includegraphics[width=.65\columnwidth]{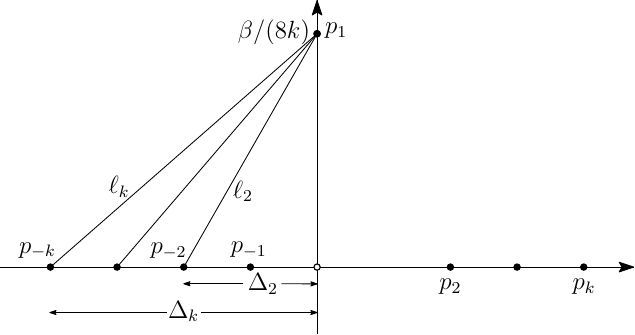}\\
	(a)\\
	\vspace{8pt}
	\includegraphics[width=.65\columnwidth]{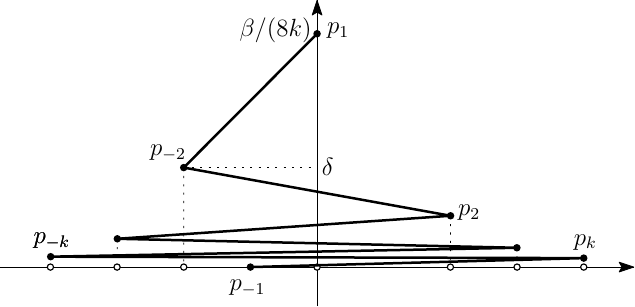}\\
	(b)
	\caption{Illustration of the construction of a longest path for $2k$ points. The figure is not \rv{true} to scale as the real $y$-coordinates are very small so that the points lie almost on the $x$-axis. (a) Lifting $(x_1,0)$ to the $y$-coordinate \rvm{$\frac{\beta}{8k}$}. (b) The final longest path.}
		\label{path-even-fig}
\end{figure}

\vspace{10pt}
\noindent{\bf Note.~~}{Figures~\ref{path-even-fig}(a) and \ref{path-even-fig}(b) are not \rv{true} to scale. The $y$-coordinates should be small enough so that all points lie almost on the $x$-axis (we exaggerated the $y$-coordinates to facilitate readability).  
}

\rvm{
\paragraph{Overview.} \revm{The rough idea of the remaining proof is as follows.} Due to the small $y$-coordinates, a longest path $H'$ on $P'$ corresponds to a \revm{longest path $H$ on $P$.} The length of $H'$ is roughly the length of $H$ plus a very small value $Y(H')$, which depends on the new $y$-coordinates. Let $e_1$ be the only edge of $H'$ incident to $p_1$. Since $p_1$ has a very large $y$-coordinate compared to other points, the contribution of $e_1$ to $Y(H')$ is larger than the contribution of other edges. The contribution of $e_1$ is maximized if it connects $p_1$ to the nearest plausible neighbor, which is $p_{-2}$; this can be observed from Figure~\ref{path-even-fig}(b).  Therefore, $e_1=p_1p_{-2}$. By a similar argument, $p_{-2}$ gets connected to $p_2$, and so on. It follows that the path $H'$ is unique and it is $p_1,p_{-2},p_2,p_{-3},p_3,\dots, p_{-k},p_k,p_{-1}.$ This path is $y$-monotone, and hence noncrossing; see Figure~\ref{path-even-fig}(b). We continue with details on how to choose the remaining $y$-coordinates.}

\rv{
\begin{lemma}\label{lem:invariant1}
If Invariant~$(M_1)$ holds, then we can choose $y_{-2}$ and $y_2$ such that $y_1>y_{-2}>y_2>0$ and Invariant~$(M_2)$ holds. 
\end{lemma}
\begin{proof}
Since Invariant~$(M_1)$ holds, every longest spanning path of $P'$ starts with $p_1$. 
We choose $y_{-2}$ and $y_{2}$ such that $y_1>y_{-2}>y_{2}>0$,
and  that every \revm{longest} spanning path of $P'$ starts with the subpath $p_1 ,p_{-2}, p_2$. 

Let $H'$ be a longest spanning path on $P'$. Notice that $H'$ corresponds to a path $H$ on $P$; denote by $\mathcal{H}$ the set of all longest spanning paths on $P$. By the triangle inequality and our constraints on the $y$-coordinates, we have $|H'|=|H|+Y(H')$ for some $Y(H')\geq 0$ that depends on the new $y$-coordinates. In particular, $Y(H')\le |P|\cdot\frac{\beta}{8k}$.

By Lemma~\ref{path-lemma}, $x_1$ is an endpoint of every path in $\mathcal{H}$. By Lemma~\ref{endpoint-lemma}(ii), the edge of $H$ incident to $x_1$ intersects the median of $P$, and so this edge is $x_1x_{-j}$ for some $j\in \{2,\ldots ,k\}$. For $j\in\{2,\dots,k\}$ let $\Delta_{j}=|x_1- x_{-j}|$ and let $\ell_j$ \revm{be} the Euclidean distance between points $p_1$ and $(x_{-j},0)$ in $\mathbb{R}^2$; as in Figure~\ref{path-even-fig}(a). 
The contribution of $p_{1}p_{-j}$ to $|H'|$ is $|p_1 p_{-j}|$\revm{.}
On the one hand, we have $|p_1 p_{-j}|\geq \ell_j - y_{-j} \geq \ell_j-y_\rev{-2}$ by the triangle inequality and the assumption $y_{-2}\geq y_{-j}$. On the other hand, we have  $|p_1 p_{-j}|\leq \ell_j$. The contribution of the corresponding edge $x_{1} x_{-j}$ to $|H|$ is exactly $\Delta_j$. Hence the contribution of $p_{1} p_{-j}$ to $Y(H')$ is at least $\ell_j-y_{\rev{-2}}-\Delta_j$ and at most $\ell_j-\Delta_j$. The triangle inequality (for the triangle $p_{1} p_{-j}p_{-(j+1)}$) yields $\ell_{j+1}<\ell_{j}+(\Delta_{j+1}-\Delta_j)$.
Consequently, we have 
\[
     \ell_{2}-	\Delta_{2}
    >\ell_{3}-\Delta_{3}
    >\dots 
    > \ell_{k}-	\Delta_{k}.
\]
    If we set $y_\rev{-2} < (\ell_{2}-\Delta_{2})-(\ell_{3}-\Delta_{3})$, 
    then the contribution of $p_{1}p_{-2}$ to $Y(H')$ is at least 
\[
    \ell_{2}-y_{\rev{-2}}-\Delta_{2}
    >\ell_{2}-\Delta_{2}-((\ell_{2}-\Delta_{2})-(\ell_{3}-\Delta_{3}))
    =\ell_{3}-\Delta_{3},
\] 
which is larger than the contribution of every other plausible edge $p_1p_{-j}$. Since the $y$-coordinates of all other points are less than $y_{\rev{-2}}$, every other edge of $H'$ contributes less than $y_{\rev{-2}}$ to $Y(H')$.
By setting 
\[
    y_\rev{-2}=\frac{(\ell_2-y_{\rev{-2}}-\Delta_2)-(\ell_{3}-\Delta_{3})}{2k-1},
\] 
the contribution of $p_{1}p_{-2}$ exceeds the sum of the contributions of the remaining $2k-2$ edges of $H'$. After rearrangement, we get
\[y_{\rev{-2}}=\frac{(\ell_2-\Delta_2)-(\ell_3-\Delta_3)}{2k}.\] 
Thus, for this choice of $y_\rev{-2}$ the longest path $H'$ contains the edge $p_1 p_{-2}$. Note also that $\ell_2<\Delta_2+y_1$ by the triangle inequality, which gives $y_\rev{-2}<\ell_2-\Delta_2<y_1$, as required. 

By Lemma~\ref{path-lemma}, $x_{-2}$ is not an endpoint of $H$.\footnote{Note that $P\setminus \{x_1\}$ contains an odd number of points, but we do not use the characterization of longest paths on odd sets (Lemma~\ref{odd-path-cor}). In fact, if we remove $x_1$ from $H$, we do not necessarily obtain a longest path on $P\setminus \{x_1\}$.}
By Lemma~\ref{endpoint-lemma}(ii), every edge of $H$ intersects the median of $P$. Consequently, $x_{-2}$ is adjacent to $x_1$ and some point $x_j$, $j\in \{2,\ldots, k\}$. Similarly to the choice of $y_\rev{-2}$, we can choose $y_\rev{2}$ so that the longest path $H'$ contains the edge $p_{-2}p_2$ and $0<y_\rev{2}<y_\rev{-2}$. This establishes Invariant~$(M_2)$.
\end{proof}
\begin{lemma}\label{lem:invariant2}
For every $i\in \{1,\ldots , k-1\}$, if Invariant~$(M_i)$ holds, then we can choose $y_{-(i+1)}$ and $y_{i+1}$ such that $y_i>y_{-(i+1)}>y_{i+1}>0$ and Invariant~$(M_{i+1})$ holds. 
\end{lemma}
\begin{proof}
The case $i=1$ is handled in Lemma~\ref{lem:invariant1}.
In general, assume that Invariant~$(M_i)$ holds for some $i\in \{2,\ldots , k-2\}$. Then the longest spanning path $H'$ on $P'$ \rev{starts with} the subpath $p_1,p_{-2},p_2,\ldots ,p_{-i}, p_i$. Let $\widehat{P}$ be the set of real numbers $x_{j}\in P$ with $|j|\geq i$; and let $\widehat{P}'$ be the corresponding subset of $P'$. Note that $\widehat{P}$ contains an even number of points, and the two closest points \revm{to} its median are $p_i$ and $p_{-1}$. 
Recall that $y_i<y_1=\frac{\beta}{8k}$.  By the choice of $\beta$, Invariant~$(M_1)$ holds for $\widehat{P}$ with the current value of $y_i$. In particular, $p_i$ and $p_{-1}$ are the endpoints of every longest spanning path on $\widehat{P}'$. Lemma~\ref{lem:invariant1} yields 
$y_{-(i+1)}$ and $y_{i+1}$ such that $y_i>y_{-(i+1)}>y_{i+1}>0$ and
every longest spanning path of $\widehat{P}'$ contains the subpath $p_i ,p_{-(i+1)}, p_{i+1}$. Consequently, every longest spanning path of $P'$ contains the subpath $p_1,p_{-2},p_2,\dots ,p_i,p_{-(i+1)},p_{i+1}$.

Finally, for $i=k-1$, we can choose $y_{-k}$ and $y_k$ arbitrarily so that  $y_{k-1}>y_{-k}>y_k>y_{-1}=0$. Invariant~$(M_{k-1})$ implies that $p_1$ and $p_{-1}$ are the endpoints of every longest spanning path $H'$ on $P'$. Therefore, $p_{k-1}$ is not an endpoint of $H'$, but its adjacent edges intersect the $y$-axis by Lemma~\ref{endpoint-lemma}(ii). 
There is only one remaining plausible edge from $p_{k-1}$ and $p_{-k}$: $p_{k-1} p_{-k}$ and $p_{-k} p_k$, respectively.  
\end{proof}

\begin{proof}[Proof of Lemma~\ref{evenpath-lem+}]
By Lemma~\ref{lem:invariant2}, Invariant~$(M_k)$ holds. This implies that there are $y$-coordinates for all points in $P'$ such that every longest path $H'$ on $P'$ corresponds to a longest path on~$P$. Furthermore, $H'$ contains $p_1, p_{-2} ,p_2,\dots ,p_{-k} ,p_k$ as a subpath. 
By Lemma~\ref{endpoint-lemma}(i), $p_1$ and $p_{-1}$ are the endpoints of $H'$. It now follows that $H' = p_1 ,p_{-2} ,p_2,\dots , p_{-k}, p_k, p_{-1}$. Since this path is $y$-monotone, it is noncrossing.
\end{proof}

\subsection{A Path with an Odd Number of Points}
\label{odd-path}

In this subsection, we extend Lemma~\ref{evenpath-lem+} to odd point sets, by reducing it to the even case. 
\begin{lemma}\label{oddpath-lem+}
 For an odd integer $n\ge 3$, let $P$ be a set of $n$ real numbers
 such that the median of $P$ has a unique closest point in $P$.
 Then there exists a set $P'$ of $n$ points in the plane with the following properties:
\begin{enumerate}
    \item the $x$-projection of $P'$ is $P$;
    \item the $x$-projection of every longest path on $P'$ is a longest path on $P$;
    \item the longest spanning path on $P'$ is unique and noncrossing.
\end{enumerate}
\end{lemma}
\begin{proof}
 Assume that $n=2k-1$ for some $k\in \mathbb{N}$ and $P=\{x_{-k},x_{-(k-1)}\ldots ,x_k\}$, where $i<j$ implies $x_i<x_j$. Note that $x_0$ is the median of $P$. We may assume without loss of generality that $x_0=0$,
 and $|x_{-1}|<|x_{1}|$ by symmetry. 
 Let $x_{1/2}=\frac12\cdot x_1$ and $Q=P\cup \{x_{1/2}\}$.
 Note that $Q$ contains an even number of points and the closest points to its median are $x_0$ and $x_{1/2}$. By Lemma~\ref{evenpath-lem+}, we can assign $y$-coordinates to the points in $Q$ to obtain a point set $Q'$ for which the longest spanning path $H'$ is unique, $y$-monotone, and its endpoints are $p_0$ and $p_{1/2}=(x_{1/2},y_{1/2})$. By the choice of $y$-coordinates, the longest path on every subset of $Q'$ corresponds to a longest path on $Q$. We may further assume that $p_{1/2}$ has the maximum $y$-coordinate in $Q'$, and $p_{1/2} p_{-1}$ is an edge of $H'$. 

Let $P'=Q'\setminus \{p_{1/2}\}$, the subset of $Q'$ corresponding to $P$.  
Since $x_{1/2}$ is an endpoint of $H'$, if we remove $x_{1/2}$ from $H'$, we obtain a spanning path $H''$ on $P'$, with endpoints $p_0$ and $p_{-1}$. We claim that $H''$ is a longest spanning path on $P'$. Suppose, for the sake of contradiction, that $P'$ admits a path $L$ that is longer than $H''$.  
By our choice of $y$-coordinates, the $x$-projection of both $H''$ and $L$ is a longest path on~$P$. By Lemma~\ref{odd-path-cor}, the endpoints of both $H''$ and $L$ are $p_{-1}$ and $p_0$. Consequently, the concatenation of $L$ and the edge $p_0 p_{1/2}$ gives a spanning path on $Q'$: its length is $|L|+|p_0 p_{1/2}|>|H'|+|p_0 p_{1/2}|=|H'|$, contradicting the maximality of $H'$. 
 \end{proof}
}

\section{Noncrossing Longest Cycles}
\label{cycle-section}
Let $n\ge 3$ be an integer. In this section, we construct a set of $n$ points for which the longest spanning cycle is unique and noncrossing. For $n=3$, every spanning cycle is noncrossng. For $n=4$, we take three vertices of a triangle and a point in the interior. Thus, we assume that $n\ge 5$. 

\begin{theorem}\label{evenpath-thm}
    For every integer $n\ge 3$ there exists a set of $n$ points in the plane for which the longest spanning cycle is unique and noncrossing. 
\end{theorem}

{\color{mycolor}In Section~\ref{even-cycle-overview} we give an overview of our construction for an even number of points. The details and proofs are given in Section~\ref{even-cycle-details}. For an odd number of points we sketch a construction in Section~\ref{odd-cycle-details}.}

\subsection{A Cycle with an Even Number of Points: An Overview}
\label{even-cycle-overview}
Let $n\geq 6$ be an even integer. Then either $n=4k$ or $n=4k{-}2$ for some integer $k\geq 2$. To simplify the indexing (of points and $y$-coordinates) in our construction, from now on we assume that $n=4k{-}2$. Let $P$ be a set of $n$ points, consisting of $2k$ points $p_i=(i,0)$ for $i=\pm1,\pm2,\dots,\pm k$ and $2k{-}2$ points $p'_i=(i{+}\epsilon,0)$ for $i=-1,\pm2,\dots,\pm (k{-}1),k$, where $\epsilon>0$ is a small value to be determined; see Figure~\ref{cycle-even-fig}. (The construction for $n=4k$ is similar; it consists of $P$ and two additional points $p_{k+1}=(k{+}1,0)$ and $p'_{-k}=(-k{+}\epsilon,0)$.) Our construction for cycles is somewhat similar to that of paths in the sense that our cycle consists of two $y$-monotone interior-disjoint paths between $p_1$ and $p_{-k}$ (or between $p_1$ and $p_{k+1}$ when $n$ is a multiple of $4$). Although the main idea sounds simple, the noncrossing property of the longest cycle is not straightforward and involves a more detailed analysis. \rev{The construction could also be done with more general $x$-coordinates as for the paths, but for simplicity we pick (almost) integers.}

{\color{mycolor}  Lemma~\ref{cycle-lemma}} implies that a spanning cycle on $P$ is longest if and only if each of its edges intersects the $y$-axis. 
Let $\cal C$ be the set of all longest spanning cycles on $P$.
As illustrated in Figure~\ref{cycle-even-fig}, we obtain a point set $P'$ by assigning to each point $p_i$ and $p'_i$ the respective $y$-coordinates $y_i$ and $y'_i$ such that:
\[{\frac{1}{16k}}=y_1\gg y_{-1}\gg y'_2\gg y_{-2}\gg y'_3\gg \cdots\gg y'_{k}\gg y_{-k}=0.\]
For each $i\in\{2, 3, \dots, k\}$ we choose $y_i$ such that $p_i$ lies just below (almost on) the segment $p'_{-i+1}p'_{i}$, and for each $i\in\{-1,-2, \dots, -(k{-}1)\}$ we choose $y'_i$ such that  $p'_i$ lies just below (almost on) the segment $p_{-i}p_{i}$.

Due to the small $y$-coordinates, \rv{every} longest cycle $C'$ on $P'$ corresponds to a cycle $C\in \cal C$. Moreover \rv{$|C'|=|C|+Y(C')$} for some small value \rv{$Y(C')$} which depends on the new $y$-coordinates. Since $p_1$ has the largest $y$-coordinate, the contribution of the two edges of $C'$ that are incident to $p_1$ (say $e_1$ and $e_2$) is maximized when they are connected to the nearest plausible neighbors which are $p_{-1}$ and $p'_{-1}$. We will choose the $y$-coordinates in such a way that the contribution of $e_1$ and $e_2$ is larger than the sum of the contributions of the remaining edges of the cycle. Thus $C'$ must connect $p_1$ to $p_{-1}$ and $p'_{-1}$. Similarly, by a suitable choice of $y$-coordinates, we enforce $C'$ to connect $p_{-1}$ and $p'_{-1}$ to the nearest plausible neighbors which are $p_2$ and $p'_2$, and so on. By repeating this process, the longest cycle $C'$ \revm{will} be the concatenation of two paths $p_1,p_{-1}, p_2, p_{-2},\dots,p_{-k}$ and $p_1,p'_{-1},p'_2,p'_{-2},\dots,p'_k,p_{-k}$.

\begin{figure}[!ht]
	\centering
	\setlength{\tabcolsep}{0in}
	\includegraphics[width=.63\columnwidth]{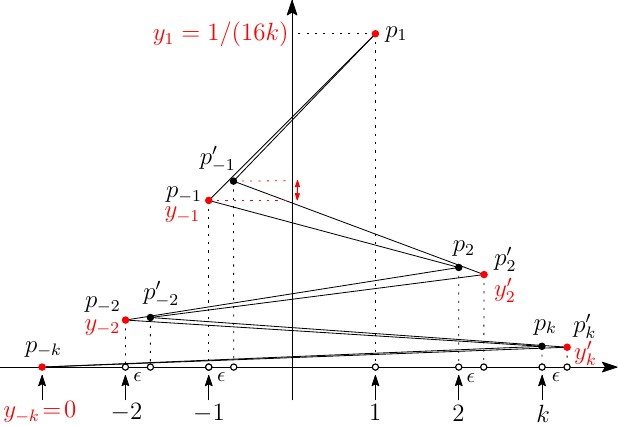}\\
	\caption{Illustration of the construction of a longest cycle for $4k{-}2$ points. The figure is not \rv{true} to scale. The $y$-coordinates should be small enough
so that all points lie almost on the $x$-axis.}
	\label{cycle-even-fig}
\end{figure}

\subsection{A Cycle with an Even Number of Points: Details}
\label{even-cycle-details}

Recall the point set $P$ from the previous \rv{subsection}  
(the $y$-coordinates and the value of $\epsilon>0$ will be determined in this section).
The longest cycles for points on a line were characterized in Lemma~\ref{cycle-lemma}. Let $\cal C$ be the set of all longest cycles on $P$.

\begin{lemma}\label{cycle-difference-lemma} 
\rv{Every} cycle in $\cal C$ is \rv{more than 2 units} longer than \rv{every} cycle not in $\cal C$.
\end{lemma}
\begin{proof}
	Consider any cycle $D$ that is not in $\cal C$.  Lemma~\ref{cycle-lemma} implies that $D$ has an edge that does not intersect the $y$-axis. Orient the edges of $D$ to make it a directed cycle. Since the number of points to the left of the $y$-axis is the same as the number of points to its right, $D$ has two directed edges $(p_{a}, p_{b})$ and $(p_c,p_d)$ such that \rv{$a,b\le -1+\epsilon$} and $c,d\ge 1$. By replacing these edges with $p_ap_c$ and $p_bp_d$ we obtain an (undirected) spanning cycle $D'$ such that \[|D'|-|D|=(|p_ap_c|+|p_bp_d|)-(|p_ap_b|+|p_cp_d|)\ge 2|p_1p'_{-1}|= 2(2-\epsilon)>2.\]Since the length of \rv{every} cycle $C$ in $\cal C$ is at least $|D'|$, we get \rv{$|C|>|D|+2$}.  
\end{proof}

To obtain $P'$ we only need to describe the following $y$-coordinates:
\[y_1\gg y_{-1}\gg y'_2\gg y_{-2}\gg y'_3\gg \cdots\gg y'_{k}\gg y_{-k}.\] 
The $y$-coordinates of the remaining points \revm{will} then follow as
outlined 
in the previous section (more details are given after Lemma~\ref{cycle-delta-epsilon-lemma}). 
We set $y_1=\frac{1}{16k}$ and $y_{-k}=0$. To assign the $y$-coordinates we use the following lemma (its proof is similar to that of Lemma~\ref{lem:invariant1}). \rev{We cannot use Lemma~\ref{lem:invariant1} and Lemma~\ref{lem:invariant2} directly because the two subpaths forming the cycle are not necessarily longest paths on their respective point sets (mainly because here $p_{-k}$ is an endpoint of the paths as opposed to $p_{-1}$).} 

\begin{lemma}
	\label{cycle-delta-lemma}
	There exists a real number $\delta$, $\epsilon\le \delta< y_1$, such that if $0\le y_i \le\delta$ for $i\neq 1$ and $0\le y'_i \le\delta$ for $i\neq -1$,
  then every longest cycle of $P'$ connects $p_1$ to $p_{-1}$ and $p'_{-1}$.
\end{lemma}
\begin{proof}
Corollary~\ref{epsilon-cor-revised} implies that \rv{every} longest cycle $C'$ on $P'$ corresponds to a cycle $C$ in $\cal C$. \rv{By the triangle inequality and our choice of the $y$-coordinates, we have $|C'|=|C|+Y(C')$ for  $0\le Y(C')\leq|P|\cdot\frac{1}{16k}\le\frac{1}{4}$.} Lemma~\ref{cycle-lemma} implies that $C'$ connects $p_1$ to two points to the left of the $y$-axis. Similar to Lemma~\ref{lem:invariant1}, for $j\in\{1,\dots,k\}$ define $\ell_{j}$ as the Euclidean distance between $p_1$ and the point $(-j,0)$, and define $\Delta_{j}$ as the difference of their $x$-coordinates. Analogously, for $j\in\{1,\dots,k{-}1\}$ define $\ell'_{j}$ and $\Delta'_{j}$ for $p_1$ and the point $(0,-j+\epsilon)$. Every edge that connects $p_1$ to a point to the left of the $y$-axis has the following contributions to $|C|$, $|C'|$ and \rv{$Y(C')$}.

\begin{itemize}
    \item For $j\in\{1,\dots,k\}$ the contribution of  $p_1p_{-j}$ to $|C'|$ is at least $\ell_j-\delta$ and at most $\ell_j$. The contribution of the corresponding edge to $|C|$ is $\Delta_j$. Hence the contribution of $p_1p_{-j}$ to \rv{$Y(C')$} is at least $\ell_j-\delta-\Delta_j$ and at most $\ell_j-\Delta_j$.
    \item For $j\in\{2,\dots,k{-}1\}$ the contribution of $p_1p'_{-j}$ to $|C'|$ is at least $\ell'_j-\delta$ and at most $\ell_j$. The contribution of the corresponding edge to $|C|$ is $\Delta'_j$. Thus the contribution of $p_1p'_{-j}$ to \rv{$Y(C')$} is at least $\ell'_j-\delta-\Delta'_j$ and at most $\ell'_j-\Delta'_j$.
    \item The contribution of $p_1p'_{-1}$ to $|C'|$ is at least $\ell'_1-\delta-\epsilon$ because the $y$-coordinate of $p'_{-1}$ is at most $\delta+\epsilon$; to verify this observe that $y_{-1}\le \delta$ and $y'_{-1}-y_{-1}<\epsilon$ because $p'_{-1}$ is almost on $p_1p_{-1}$ whose slope is less than $1$; also see Figure~\ref{cycle-even-fig} (recall that the figure is not to \rv{true} scale). The contribution of the corresponding edge to $|C|$ is $\Delta'_{1}$. Therefore the contribution of $p_1p'_{-1}$ to \rv{$Y(C')$} is at least $\ell'_1-\delta-\epsilon-\Delta'_1$ and at most $\ell'_1-\Delta'_1$. 
\end{itemize}

Observe that 
	\[\ell'_{1}-
	\Delta'_{1}>\ell_{1}-
	\Delta_{1}>\ell'_{2}-
	\Delta'_{2}>\ell_{2}-
	\Delta_{2}> \dots > \ell_{k}-
	\Delta_{k}.\]
If we set $\delta <\frac{1}{2}\left((\ell_{1}-
	\Delta_{1})-(\ell'_{2}-
	\Delta'_{2})\right)$, then the contributions of $p_1p_{-1}$ and $p_1p'_{-1}$ to \rv{$Y(C')$} would respectively be at least 
 \[\ell_{1}-\delta-	\Delta_{1}>\ell_{1}-2\delta-
	\Delta_{1}>\ell_{1}-
	\Delta_{1}-((\ell_{1}-
	\Delta_{1})-(\ell'_{2}-	\Delta'_{2}))=\ell'_{2}-
	\Delta'_{2},\textrm{ and}\] \[ \ell'_1-\delta-\epsilon-\Delta'_1\ge\ell'_{1}-2\delta-	\Delta'_{1}>\ell'_{1}-
	\Delta'_{1}-((\ell_{1}-
	\Delta_{1})-(\ell'_{2}-
	\Delta'_{2}))>\ell'_{2}-
	\Delta'_{2},\] which are larger than the contribution of \rv{every} other edge $p_1p_{-j}$ and $p_1p'_{-j}$. 
	By setting    
    \rvm{\[
    \delta=\frac{1}{2}\cdot\min\left\{\frac{(\ell_{1}-\delta-
	\Delta_{1})-(\ell'_{2}-
	\Delta'_{2})}{4k-2}, \frac{(\ell'_{1}-\delta-\epsilon-
	\Delta'_{1})-(\ell'_{2}-
	\Delta'_{2})}{4k-2}\right\},\]}the contribution of $p_1p_{-1}$ and $p_1p'_{-1}$ each would be even larger than the sum of the contributions of the remaining $4k{-}4$ edges of $C'$. \rvm{Since $\epsilon\leq \delta$, after rearranging we \revm{finally fix} \[\delta=\min\left\{\frac{(\ell_{1}-
	\Delta_{1})-(\ell'_{2}-
	\Delta'_{2})}{8k-3}, \frac{(\ell'_{1}-
	\Delta'_{1})-(\ell'_{2}-
	\Delta'_{2})}{8k-2}\right\}.\]}
 Thus, for this choice of $\delta$, the longest cycle $C'$ connects $p_1$ to $p_{-1}$ and $p'_{-1}$.
\end{proof}

We choose $\delta$ as in the proof of Lemma~\ref{cycle-delta-lemma}, and set $y_{-1}=\delta$. Then we set  $y'_{-1}$ so that $p'_{-1}$ lies just below (almost on) the segment $p_1p_{-1}$, as in Figure~\ref{cycle-even-fig}. {\color{mycolor}Notice that $\delta<y'_{-1}<\delta+\epsilon=y_{-1}+\epsilon$.} Then, by Lemma~\ref{cycle-delta-lemma} the longest cycle connects $p_1$ to $p_{-1}$ and $p'_{-1}$. By Lemma~\ref{cycle-lemma}, the other edges incident to $p_{-1}$ and $p'_{-1}$ must cross the $y$-axis.

{\color{mycolor}
\begin{lemma}
\label{cycle-delta-epsilon-lemma}
    There exists a real \revm{number} $\delta$, $\epsilon \le \delta<y_{-1}$, such that if $0\le y_i \le\delta$ for $i\neq -1,1,2$ and $0\le y'_i \le\delta$ for $i\neq -1$,
  then every longest cycle of $P'$ connects $p_{-1}$ to $p_2$ and $p'_{-1}$ to $p'_{2}$.
\end{lemma}
\begin{proof}
Recall the longest cycle $C'$ from the proof of Lemma~\ref{cycle-delta-lemma}. We \revm{will} choose $\delta$ small enough such that the contribution of each of $p_{-1}p_2$, $p_{-1}p'_2$, $p'_{-1}p_2$, and $p'_{-1}p'_2$ to \rv{$Y(C')$} is larger than the sum of the contributions
of the remaining $4k{-}6$ edges of $C'$. This \revm{will} force $C'$ to connect $p_{-1}$ and $p'_{-1}$ to $p_{2}$ and $p'_{2}$.

By an argument similar to that of Lemma~\ref{cycle-delta-lemma} we can find a parameter $\delta_{1}$ that forces $C'$ to connect $p_{-1}$ to $p_2$ or $p'_2$ ($\delta_1$, $y_{-1}$, $p_{-1}$, $p_2$, and $p'_2$ play the roles of $\delta$, $y_1$, $p_1$, $p'_{-1}$, and $p_{-1}$, respectively). Similarly, we can find a parameter $\delta'_{1}$  that forces $C'$ to connect $p'_{-1}$ to $p_2$ or $p'_2$ (where $\delta'_1$, $y'_{-1}$, $p'_{-1}$, $p_2$, and $p'_2$ play the roles of $\delta$, $y_1$, $p_1$, $p'_{-1}$, and $p_{-1}$, respectively).
Then we choose $\delta=\min\{\delta_{1},\delta'_{1}\}$.

Our choice of $\delta$ ensures that $C'$ connects $p_{-1}$ and $p'_{-1}$ to $p_{2}$ and $p'_{2}$. Notice that $p_{-1}$ and $p'_{-1}$ cannot both connect to $p_{2}$ or to $p'_{2}$ because it closes the cycle. Thus $C'$ must use $p_{-1}p_2$ and $p'_{-1}p'_{2}$ or $p_{-1}p'_{2}$ and $p'_{-1}p_{2}$. We show that $C'$ uses $p_{-1}p_2$ and $p'_{-1}p'_{2}$. See Figure~\ref{cycle-edges-fig}.
Recall that $p_{2}$ is almost on the edge $p'_{-1}p'_{2}$, and hence $|p'_{-1}p'_{2}|\approx |p'_{-1}p_{2}|+|p_{2}p'_{2}|$. By the triangle inequality we get $|p_{-1}p_{2}|+|p_{2}p'_{2}|>|p_{-1}p'_{2}|$.
Adding these two yields 
\begin{equation}\label{eq1}|p_{-1}p_{2}|+|p'_{-1}p'_{2}|> |p_{-1}p'_{2}|+|p'_{-1}p_{2}|,\end{equation}
which means that $C'$ connects $p_{-1}$ to $p_2$ and $p'_{-1}$ to $p'_{2}$.  
\end{proof}
\begin{figure}[H]
	\centering
\setlength{\tabcolsep}{0in}
\includegraphics[width=.47\columnwidth]{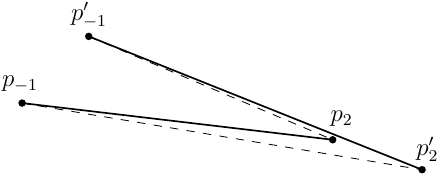}\\
	\caption{The longest cycle connects $p_{-1}$ to $p_{2}$ and $p'_{-1}$ to $p'_{2}$}
	\label{cycle-edges-fig}
\end{figure}

We choose our new $\delta$ as in the proof of Lemma~\ref{cycle-delta-epsilon-lemma}, and set $y'_{2}=\delta$. Now that the point $p'_2$ is fixed we can choose the $y$-coordinate of $p_2$ in the triangle $\bigtriangleup p_{-1}p'_{-1}p'_2$ and very close to the segment $p'_{-1}p'_2$ such that~\eqref{eq1} holds.
This forces the longest cycle to use $p_{-1}p_2$ and $p'_{-1}p'_2$. \rev{We can repeat the argument of  Lemma~\ref{cycle-delta-epsilon-lemma} such that for $j\ge 2$ the longest cycle will use the edges $p_jp_{-j}$, $p'_jp'_{-j}$, $p_{-j}p_{j+1}$, and $p'_{-j}p'_{j+1}$}. Therefore
the longest cycle on $P'$ is the concatenation of two paths: $p_1,p_{-1}, p_2, p_{-2},\dots,p_{-k}$ and $p_1,p'_{-1},p'_2,p'_{-2},\dots,p'_k,p_{-k}$. This cycle is unique and noncrossing. 

Each time we apply Lemma~\ref{cycle-delta-epsilon-lemma} we obtain a new value for $\delta$. In each \rv{iteration} we need $\delta$ to be greater than or equal to our fixed parameter $\epsilon$. For this purpose, we choose $\epsilon$ to be the parameter $\delta$ that is obtained in the last invocation of Lemma~\ref{cycle-delta-epsilon-lemma}, that is\revm{,} $\epsilon=y'_k$.

\subsection{A Cycle with an Odd Number of Points: An Overview}
\label{odd-cycle-details}

We start with an outline of our construction; see Figure~\ref{cycle-odd-fig} for an illustration. Let $n=2k{+}1$, for $k\geq 2$. We choose a set of $x$-coordinates as $P=\{-k, -(k{-}1)\epsilon,-(k{-}2)\epsilon,\ldots , -\epsilon,0, 1,2,\ldots  ,k\}$, where $\epsilon\in (0,1/(16k^2))$ will be specified later. Note that $0$ is the median of $P$, and the set $A=\{-i\cdot \epsilon: i=0,1,\ldots, k-1\}\subset [-1/(16k),0]$ forms a small cluster. 

\begin{figure}[!ht]
	\centering
	\setlength{\tabcolsep}{0in}
	\includegraphics[width=.75\columnwidth]{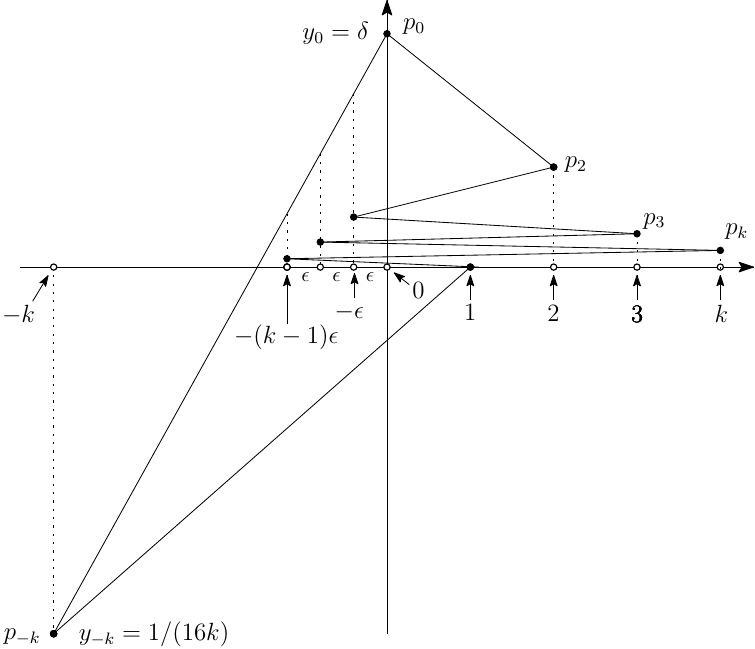}\\
	\caption{Illustration of the construction of a longest cycle for $2k{+}1$ points. }
	\label{cycle-odd-fig}
\end{figure}

Below, we will specify a $y$-coordinate for each element in $P$. This will result in the point set $P'$ for which the longest spanning cycle is unique and noncrossing. We will denote by $A'$ the set of points in $P'$ corresponding to $A$. Let $\mathcal{C}$ be the set of longest spanning cycles on $P'$.

It remains to specify the $y$-coordinates of the points in $P'$ and the parameter $\epsilon$. Let $p_x$ denote the point in $P'$ with $x$-coordinate $x\in P$. We first choose the $y$-coordinate for the leftmost point: Let $y_{-k}=-1/(16k)$; this is the only negative $y$-coordinate. We assume that $|y_i|\ll 1/(16k)$ for all other points.
\rev{We cannot use Corollary~\ref{epsilon-cor-revised} directly for $P$ because} $1/(16k)$ is too large to make any conclusion about the longest spanning cycles as the points in cluster $A'$ are indistinguishable.} \rev{Thus we define a point set, close enough to $P$, for which we can apply Corollary~\ref{epsilon-cor-revised}.} Let $\widehat{P}$ denote the multiset obtained from $P$ \revm{by} perturbing all $k$ points in $A$ to 0. Let $\mathcal{A}$ be the set of longest spanning cycles on $\widehat{P}$, and $\mathcal{B}$ the set of all other spanning cycles on $\widehat{P}$. Since the points in $\widehat{P}$ are integers, then $\beta\geq 1$, where $\beta= \min\{|A|-|B| : A\in \mathcal{A}, B\in \mathcal{B}\}$. Corollary~\ref{epsilon-cor-revised} implies that every longest spanning cycle $C'$ on $P'$ corresponds to \revm{a} cycle $C\in \mathcal{A}$. This, in turn, implies that $|C'|=|C|+Y(C',\epsilon)$ for some small value $Y(C',\epsilon)$ that depends on the new $y$-coordinates and $\epsilon$.

\begin{lemma}\label{lem:pkneighbors}
    In every longest cycle on $P'$, point $p_{-k}$ is adjacent to at most one point in cluster $A'$ and at least one point in $\{p_1,\ldots , p_k\}$. 
\end{lemma}
\begin{proof}
Let $C'$ be a longest spanning cycle on $P'$, corresponding to a longest spanning cycle $\widehat{C}$ on~$\widehat{P}$. By Lemma~\ref{cycle-lemma}(i), every edge of $\widehat{C}$ intersects the median $0$ of $\widehat{P}$. This implies that every edge of $C'$ intersects the $y$-axis or has an endpoint in the cluster $A'$.
    
Note that both $P$ and $\widehat{P}$ contain the unit-length interval $I=[0,1]$ between the leftmost $k{+}1$ and the rightmost $k$ points. By Lemma~\ref{cycle-lemma}(iii), $n-1$ edges of $\widehat{C}$ contain the interval $[0,1]$. This means that at most one edge of $\widehat{C}$ does not contain $[0,1]$. Consequently, in the cycle $C'$, at most one edge \revm{connects} $p_{-k}$ to a point in $A'$ and at least one edge connects it to a point in $\{p_1,\ldots , p_k\}$. 
\end{proof}

Similar to Lemma~\ref{cycle-delta-lemma}, there is a threshold $\delta>0$ such that if $0\leq y_i\leq \delta$ for all $i\neq -k$, then in every longest cycle on $P'$, point $p_{-k}$ is adjacent to the two closest plausible points (subject to the constraints in Lemma~\ref{lem:pkneighbors}): That is, $p_{-k}$ is adjacent to a point in cluster $A'$ and to $p_{1}$.
We set $y_0=\delta$ and find a threshold $\delta_1\in (0,\delta)$ such that if $0\leq y_i\leq \delta_1$ for all remaining points and $0<\epsilon<\delta_1$, then \rv{the contribution of $p_{-k}p_0$ and $p_{-k}p_1$ to $Y(C',\epsilon)$ each} exceeds the sum of contributions of all remaining edges of a cycle $C'\in \mathcal{C}$. Consequently, every longest cycle on $P'$ must include \rv{the edges $p_{-k}p_0$ and $p_{-k}p_1$}. Now both $p_{-k}$ and $p_0$ are fixed, and we choose a sufficiently small $\epsilon\in (0,\delta_1)$ such that all remaining points in the cluster $A'$ are below $p_{-k}p_0$ for all possible $y$-coordinates, \rev{and that distances of the points in $A'$ to $0$ are smaller than $1/(16k)$.}

\rv{Since $\epsilon>0$ is fixed now, $Y(C',\epsilon)$ depends only on the $y$-coordinates that are still unspecified.}
A longest cycle $C'$ on $P'$ comprises of $p_{-k}p_0$, $p_{-k}p_1$, and a longest spanning path $H'$ on $P'\setminus \{p_{-k}\}$ from $p_0$ to $p_1$ on $P'\setminus \{p_{-k}\}$.
\rv{Let $\mathcal{H}$ denote the set of such paths $H'$. Note that $H'$ need not be a longest path on $P'\setminus \{p_{-k}\}$ (in particular, $y_0=\delta$ may be too large to apply Corollary~\ref{epsilon-cor-revised} for $P'\setminus \{p_{-k}\}$). However, since $C'$ corresponds to a cycle $C\in \mathcal{A}$, then all edges in $H'$ intersect the line $y=\frac12$ (which is the median of $P'\setminus \{p_{-k}\}$). 
By Lemma~\ref{path-lemma}, $H'$ corresponds to a longest spanning path on $P\setminus \{-k\}$. As $P\setminus \{-k\}$ contains an even number of points. We can now proceed as in the proof of Lemma~\ref{evenpath-lem+},} and successively choose $y$-coordinates for the remaining points such that $H'$ is unique and $y$-monotone (hence noncrossing); and $y_1=0$. 

Finally, we show that $C'$ is also noncrossing. Edge $p_{-k}p_1$ lies below the $x$-axis, hence below the entire path $H'$; and $P'\setminus \{p_{-k},p_0\}$ lies below the supporting line of $p_{-k}p_0$. Consequently, the concatenation of $p_{-k}p_0$, $p_{-k}p_1$ and $H'$ is noncrossing. 

\begin{remark}{\rm Our construction in this section suggests an alternative construction for even cycles that can be obtained by connecting a point to both endpoints of an odd path.}
\end{remark}

\section{Noncrossing Longest Matchings}
\label{matching-section}
\rv{An example of a point set} for which the longest perfect matching is noncrossing is already known~\cite{Rebollar2024}. 
This example is attributed to K\r{a}ra P.~Villanger in a paper by Tverberg~\cite{Tverberg1979}.   As illustrated in Figure~\ref{Villanger-fig}, it consists of a set $S$ of $k$ segments with endpoints in $A=\{a_1,\dots,a_k\}$ and $B=\{b_1,\dots,b_k\}$. The distance between \rv{every} two points $a_i\in A$ and $b_j\in B$ is larger than the distance between \rv{every} two points in $A$, or the distance between \rv{every} two points in $B$. The points in $B$ are roughly on a vertical line. A precise description of the construction along with a detailed proof that $S$ is a longest matching for $A\cup B$ \rv{is provided in~\cite{Rebollar2024}}.
\begin{figure}[H]
	\centering	
 \vspace{-7pt}
\includegraphics[width=.6\columnwidth]{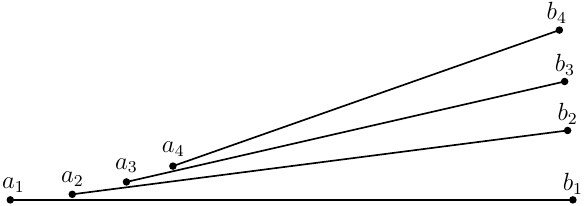}
	\caption{Villanger’s configuration as illustrated in \cite{Rebollar2024}.}
	\label{Villanger-fig}
\end{figure}

Here, we exhibit an alternative point set for which the longest perfect matching is noncrossing.
Our construction follows the same framework as for paths and cycles. \rev{For simplicity we pick integer $x$-coordinates, however, it can be done with more general $x$-coordinates.}
Let $P$ be a set of $2k$ points $p_i=(i,0)$ for $ i=\pm1,\pm2,\dots,\pm k$. One can verify that a perfect matching on $P$ is longest if and only if all edges cross the $y$-axis. One such matching is $M=\{p_{-i}p_i\colon i=1,\dots,k\}$. Using ideas similar to those used for paths and cycles, one can assign to each $p_i$ a new $y$-coordinate $y_i$ to make $M$ longest and noncrossing at the same time; see Figure~\ref{matching-fig}. The new $y$-coordinates are of the following form: $y_1\gg y_{-1}=y_2\gg y_{-2}=y_3\gg\cdots= y_{k}\gg y_{-k}$.

\begin{figure}[!ht]
	\centering
\includegraphics[width=.6\columnwidth]{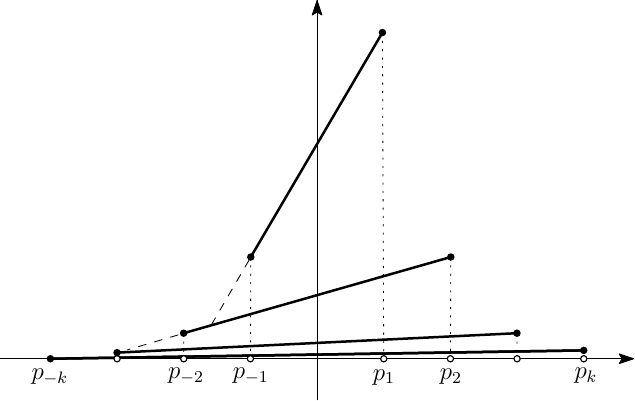}
	\caption{Illustration of our construction of a longest matching.}
		\label{matching-fig}
\end{figure}

\section{Some Properties of Longest Paths and Cycles}
\label{properties-sec}

In this section we give some structural properties of longest paths and cycles, 
possibly of independent interest. We state these properties only for cycles, but they hold for paths as well.
Two edges are in \emph{convex position} if they are edges of their convex hull. Two directed edges in convex position have {\em the same orientation} if they are both directed clockwise or counterclockwise along their convex hull.

\begin{observation}
\label{inconsistent-obs}
	Suppose that we orient the edges of a longest cycle $C$ to make it a directed cycle. Then $C$ cannot have \rv{any} pair of nonadjacent edges that are in convex position and have the same orientation along their convex hull.
\end{observation}

To verify this, note that if $C$ has two such edges, say $e_1$ and $e_2$, then flipping them (replacing $e_1$ and $e_2$ by the two diagonals of the convex hull of $e_1$ and $e_2$) would produce a longer undirected cycle as in Figure~\ref{inconsistent-fig}(a). Since $e_1$ and $e_2$ have the same orientation along their convex hull, the flip does not break the cycle into two components. 
{\color{mycolor}
If every directed simple polygon $S$ contained a pair of nonadjacent edges in convex position with the same orientation along their convex hull, Observation~\ref{inconsistent-obs} would imply Conjecture~\ref{cycle-conjecture}. However, some simple polygons do not have edges that can be flipped in this way; see \rv{for example}, Figure~\ref{inconsistent-fig}(b).}

\begin{figure}[ht!]
	\centering
	\setlength{\tabcolsep}{0in}
	$\begin{tabular}{cc}
		\multicolumn{1}{m{.5\columnwidth}}{\centering\vspace{0pt}\includegraphics[width=.37\columnwidth]{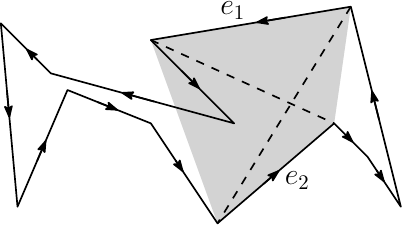}}
		&\multicolumn{1}{m{.5\columnwidth}}{\centering\vspace{0pt}\includegraphics[width=.45\columnwidth]{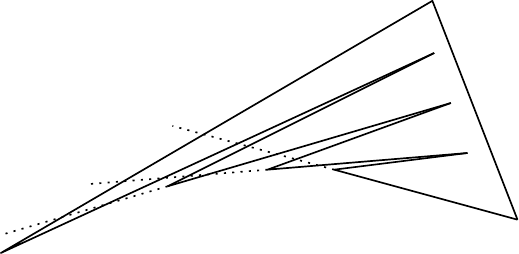}}
		\\
		(a)   &(b) 
	\end{tabular}$	
	\caption{(a) Flipping two edges in convex position. (b) A simple polygon with no pair of edges in convex position that  have the same orientation, no matter how we direct the polygon.}
\label{inconsistent-fig}
\end{figure}

\begin{observation}\label{diam-obs}
The longest cycle need not contain an edge between diametric points.
\end{observation}

{\color{mycolor}To verify this observation consider an isosceles right triangle $abc$ whose right angle is at $b$.
Place one point at $a$, one point at $c$, and two or more points very close to $b$. Then, the longest cycle does not contain the diametric point pair $\{a,c\}$. 
This observation implies that a longest cycle may not be achieved by greedily choosing longest edges.  

The following proposition implies that if the longest cycle is noncrossing, it contains some edge whose length is among the smallest three-quarters of all distances defined by its vertices.}

{\color{mycolor}
\begin{proposition}
Let $S$ be a simple polygon (a noncrossing cycle) on $n$ points. Then \rv{the length of the shortest edge of $S$} is among the smallest \rv{$3n^2/8-7n/8+1$} distances of the $\binom{n}{2}$ point pairs.  
\end{proposition}

\begin{proof}
Let $e$ and $e'$ be two edges of $S$ such that their distance along $S$ (\rv{the number of edges between them}) 
is at least 2. Since $S$ is a simple polygon, $e$ and $e'$ do not cross. Thus, there is an endpoint $p$ of $e$ and an endpoint $p'$ of $e'$ such that $|pp'|$ is larger than the length of the shorter of $e$ and $e'$, and $pp'$ is not an edge of $S$. The number of pairs of edges at distance $0$ is $n$, and the number of pairs of edges at distance $1$ is also $n$. Thus, the total number of pairs of edges at a distance at least $2$ is $\binom{n}{2}-2n$. Each such pair of edges yields a pair $\{p,p'\}$. Each $\{p,p'\}$ can be counted for $4$ different pairs of edges \rv{along $S$} that are obtained by combining the two edges incident to $p$ and the two edges incident to $p'$. Therefore, the total number of distinct pairs $\{p,p'\}$ is at least $\frac{1}{4}\left(\binom{n}{2}-2n\right)$. Subtracting this from the total number $\binom{n}{2}$ of point pairs \rv{and then excluding the $n-1$ longest edges of $S$ yield the claimed bound for the shortest edge of $S$}.
\end{proof}

\paragraph{Acknowledgements.}{This work was initiated at the 10th Annual Workshop on Geometry and Graphs, held at Bellairs Research Institute in Barbados in February 2023. We thank the organizers and the participants. We are very grateful to the anonymous reviewers who meticulously
verified our proofs, and provided detailed feedback that helped to generalize main arguments in our proofs, and to improve the readability of the paper.}

\bibliographystyle{plainurl}
\bibliography{Noncrossing-Long-Cycles}

\end{document}